\def\year{2020}\relax
\documentclass[letterpaper]{article} 
\usepackage{aaai20}  
\usepackage{times}  
\usepackage{helvet} 
\usepackage{courier}  
\usepackage[hyphens]{url}  
\usepackage{graphicx} 
\usepackage{graphicx}  
\usepackage{amsmath}
\usepackage{amsthm}

\newtheorem{lemma}{Lemma}
\newtheorem{thm}{Theorem}

\newcommand{\denselist}{\itemsep 0pt\partopsep 0pt}

\newenvironment{proofsketch}{%
  \proof}{\endproof}

\theoremstyle{definition}

\usepackage{pbox}
\usepackage{enumerate}
\usepackage{multirow}
\usepackage{booktabs}
\usepackage{amsfonts}
\usepackage{gensymb}
\usepackage{bm}
\usepackage{subcaption}
\usepackage{accents}

\newcommand{\RR}{I\!\!R} 

\DeclareMathOperator*{\argmax}{arg\!max}

\title{End-to-End Game-Focused Learning of Adversary Behavior in Security Games}
\author{
Andrew Perrault,\textsuperscript{\rm 1}
Bryan Wilder,\textsuperscript{\rm 1}
Eric Ewing,\textsuperscript{\rm 2}
Aditya Mate,\textsuperscript{\rm 1}
Bistra Dilkina,\textsuperscript{\rm 2}
Milind Tambe\textsuperscript{\rm 1}\\
\textsuperscript{\rm 1}Center for Research on Computation and Society, Harvard\\
\textsuperscript{\rm 2}Center for Artificial Intelligence in Society, University of Southern California\\
aperrault@g.harvard.edu, bwilder@g.harvard.edu, ericewin@usc.edu, \\ aditya\_mate@g.harvard.edu,
dilkina@usc.edu, milind\_tambe@harvard.edu
}
\begin{document}
\maketitle
\begin{abstract}
Stackelberg security games are a critical tool for maximizing the utility of limited defense resources to protect important targets from an intelligent adversary. 
Motivated by green security, where the defender may only observe an adversary's response to defense on a limited set of targets, we study the problem of learning a defense that \emph{generalizes} well to a new set of targets with novel feature values and combinations. Traditionally, this problem has been addressed via a \emph{two-stage} approach where an adversary model is trained to maximize predictive accuracy without considering the defender's optimization problem. We develop an end-to-end \emph{game-focused} approach, where the adversary model is trained to maximize a surrogate for the defender's expected utility. We show both in theory and experimental results that our game-focused approach achieves higher defender expected utility than the two-stage alternative when there is limited data.
\end{abstract}

\section{Introduction}
Many real-world settings call for allocating limited defender resources against a strategic adversary, such as protecting public infrastructure \cite{Gan2015SecurityGW}, transportation networks \cite{okamoto2012solving}, large public events \cite{yin2014game}, urban crime \cite{zhang2015keeping}, and green security \cite{fang2015security}.
\emph{Stackelberg security games (SSGs)} are a critical framework for computing defender strategies that maximize expected defender utility to protect important targets from an intelligent adversary \cite{tambe2011security}.

In many SSG settings, the adversary's preferences over targets are not known a priori. In early work, the adversary's preferences were estimated via the judgments of human experts \cite{tambe2011security}. In domains where there are many interactions with the adversary, we can leverage this history using machine learning instead. This line of work, started by Letchford et al.~\shortcite{Letchford2009LearningAA}, has received extensive attention in recent years (see related work).

We use protecting wildlife from poaching \cite{fang2015security} as a motivating example. The adversary's (poacher's) behavior is observable because snares are left behind, which rangers aim to remove (see Fig.~\ref{fig:snares}). Various features such as animal counts, distance to the edge of the park, weather and time of year may affect how attractive a particular target is to the adversary.
The training data consists of adversary behavior in the context of particular sets of targets, and our objective is to achieve a high defender utility when we are playing against the same adversary and new sets of targets.
For the problem of poaching prevention, Gholami et al.~\shortcite{gholami2018adversary} use around 20 features per target and observe tens of thousands of distinct targets (i.e., combinations of feature values). Rangers patrol a small portion of the park each day and aim to predict poacher behavior across a large park consisting of targets with novel feature values.

The standard approach to the problem breaks it into two stages. In the first, the adversary model is fit to the historical data to minimize an accuracy-based loss function, and in the second, the defender covers the targets (via a mixed strategy) to maximize utility against the learned model. 
It is true that, in a worst-case analysis, a model that is more accurate in a global sense induces a better coverage (see Sinha et al.~\shortcite{sinha2016learning} and Haghtalab et al.~\shortcite{haghtalab2016three}), but a model that accurately predicts the \emph{relative} values of ``important'' targets may achieve high defender utility with weak global accuracy. For example, in a game with many low-value targets, the estimates of the values of the low-value targets can be wildly inaccurate and still yield a high defender utility (see Sec.~\ref{sec:twostageimpact} for an example).

In our \emph{game-focused} approach, in contrast to a two-stage approach, we focus on learning a model that yields a high defender expected utility from the start. We train a predictive model end-to-end (i.e., considering the effects of the optimization problem) using an estimate of defender expected utility as our loss function. This approach has the advantage of focusing learning on ``important'' targets that have a large impact on the defender expected utility, and not being distracted by irrelevant targets (e.g., those with low value for both the attacker and defender). For example, in our human subject data experiments, two-stage achieves 2--20\% lower cross entropy, but worse defender expected utility.
Performing game-focused training requires us to overcome several technical challenges, including forming counterfactual estimates of the defender's expected utility and differentiating through the solution of a nonconvex optimization problem.

In summary, our contributions are: \emph{First}, we provide a theoretical justification for why our game-focused approach can outperform two-stage approaches in SSGs.
\emph{Second}, we overcome technical challenges to develop a game-focused learning pipeline for SSGs.
\emph{Third}, we test our approach on a combination of synthetic and human subject data and show that game-focused learning outperforms a two-stage approach in settings where the amount of data available is small and when there is wide variation in the adversary's values for the targets.

\paragraph{Related Work.} There is a rich literature on SSGs, ranging from uncertain observability  \cite{korzhyk2011solving} to disguised defender resources \cite{guo2017comparing} to extensive-form models \cite{cermak2016using} to patrolling on graphs \cite{basilico2012patrolling,basilico2017adversarial}. In particular, learning to maximize the defender's payoff from repeated play has been a subject of extensive study. It is important to distinguish between the active learning case \cite{Letchford2009LearningAA,xu2016playing,blum_haghtalab_procaccia_2017}, where the defender may gather information through her choice of strategy, and the passive case, where the defender does not have control over the training data. We consider each case to be valuable but focus on the passive case because we believe it is encountered more frequently in domains of interest. In the anti-poaching setting, parks often have historical data that far exceeds what can be actively collected in the short term.

Bounded rationality models are a critical component of the SSG literature because they allow the defender to achieve higher utilities against many realistic attackers. They have been the subject of extensive study since their introduction by Pita et al.~\shortcite{Pita2010RobustST} (e.g., Cui and John~\shortcite{cui2014empirical} and Abbasi et al.~\shortcite{abbasiadversaries}, who develop a distinct line of work, inspired by psychology).
We focus on the \emph{quantal response (QR)} \cite{YANG2013440} model and especially the \emph{subjective utility quantal response (SUQR)} model \cite{nguyen2013analyzing}. SUQR is simple, widely used and has been shown to be effective in practice.

Sinha et al.~\shortcite{sinha2016learning} and Haghtalab et al.~\shortcite{haghtalab2016three} provide probabilistic bounds on the learning error for two-stage approaches for generalized SUQR attackers in the passive and weakly active cases, respectively. Both works translate these bounds into the guarantees on the defender's expected utility in the worst case. Our focus is on the orthogonal issue of how to train \emph{any} differentiable predictive model end-to-end with gradient descent, including deep learning architectures that are the state of the art for many learning tasks. These methods can scale to many features and complicated relationships and are one of the main appeals of two-stage approaches. We use SUQR implemented on a neural network as an illustrative example, but our approach can be applied to other bounded rationality models, as we discuss in Sec.~\ref{sec:decision}.


Outside of SSGs, Ling et al.~\shortcite{ling2018game,ling2019game} use a differentiable QR equilibrium solver to reconstruct the payoffs of both players in a game from observed play. Hartford et al.~\shortcite{hartford2016deep} and Wright and Leyton-Brown~\shortcite{WRIGHT201716} study the problem of predicting play in unseen two-player simultaneous-move games with a small number of actions per player, and Hartford et al.~\shortcite{hartford2016deep} build a deep learning architecture for this purpose. These works focus on prediction rather than optimization.

We briefly discuss related work in end-to-end learning for decision-making in non-game-theoretic contexts (see Donti and Kolter~\shortcite{donti2017task} for a more complete discussion). New technical issues arise due to the presence of the adversary, such as counterfactual estimation and nonconvexity.
In their study of parameter sensitivity, Rockafellar and Wets~\shortcite{rockafellar2009variational} provide a comprehensive theoretical analysis of differentiating through optimization. Bengio~\shortcite{Bengio1997UsingAF} was first to train a learning system for a more complex task by directly differentiating through the outcome of applying parameterized rules. Amos and Kolter~\shortcite{amos2017optnet} provide analytical derivatives for constrained convex problems. This analytic approach is extended to stochastic optimization by Donti et al.~\shortcite{donti2017task} and to submodular optimization by Wilder et al.~\shortcite{wilder2019melding}. Demirovic et al.~\shortcite{Demirovic2019PredictOptimiseWR} provide a theoretically optimal framework for ranking problems with linear objectives.

\section{Setting}
\label{sec:setting}
\paragraph{Stackelberg Security Games (SSGs).}
Our focus is on optimizing defender strategies for SSGs, which describe the problem of protecting a set of targets given limited defense resources and constraints on how the resources may be deployed \cite{tambe2011security}. Formally, an SSG is a tuple $\{\mathcal{T}, \bm{u}_d, \bm{u}_a, C_d\}$, where $\mathcal{T}$ is a set of targets, $\bm{u}_d: \mathcal{T} \rightarrow \RR_{\leq 0}$ is the defender's payoff if each target is successfully attacked, $\bm{u}_a: \mathcal{T} \rightarrow \RR_{\geq 0}$ is the attacker's, and $C_d$ is the set of constraints the defender's strategy must satisfy. Both players receive a payoff of zero when the attacker attacks a target that is defended.

The game has two time steps: the defender computes a mixed strategy that satisfies the constraints $C_d$, which induces a \emph{marginal coverage probability (or coverage)} $\bm{p}=\{ \bm{p}_i: i \in \mathcal{T} \}$. The attacker's \emph{attack function} $\bm{q}$ determines which target is attacked, inducing an \emph{attack probability} for each target. The defender seeks to maximize her expected utility:
\begin{small}
\begin{align}
    \label{eq:defenders}
    &\max_{\bm{p} \textrm{ satisfying } C_d} \textsc{DEU}(\bm{p}; \bm{q})=\notag \\ &\max_{\bm{p} \textrm{ satisfying } C_d}\sum_{i \in \mathcal{T}} (1-\bm{p}_i) \bm{q}_i(\bm{u}_a, \bm{p}) \bm{u}_d(i).
\end{align}
\end{small}
The attacker's $q$ function can represent a rational attacker, e.g., $\bm{q}_i(\bm{p}, \bm{u}_a) = 1 \textrm{ if } i=\argmax_{j \in \mathcal{T}}(1-\bm{p}_j)\bm{u}_a(j) \textrm{ else } 0$, or a boundedly rational attacker. 
A QR attacker \cite{mckelvey1995quantal} attacks each target with probability proportional to the exponential of its payoff scaled by a constant $\lambda$, i.e., $\bm{q}_i(\bm{p}) \propto \exp(\lambda(1-\bm{p}_i)\bm{u}_a)$. An SUQR \cite{nguyen2013analyzing} attacker attacks each target with probability proportional to the exponential of an \emph{attractiveness function}:
\begin{small}
\begin{align}\label{eq:suqr}
    \bm{q}_i(\bm{p}, \bm{y}) \propto \exp(w\bm{p}_i + \phi(\bm{y}_i)),
\end{align}
\end{small}where $\bm{y}_i$ is a vector of features of target $i$ and $w<0$ is a constant. We call $\phi$ the \emph{target value function}. We focus our effort on learning $\phi$ because $w$ can easily be learned using existing techniques, such as the maximum likelihood estimation (MLE) approach of Sinha et al.~\shortcite{sinha2016learning}, assuming we have the ability to play different defender strategies against the same set of targets. MLE estimates converge rapidly, as shown by Fig.~\ref{fig:w_convergence}, which demonstrates learning in an eight-target game, averaged over 20 trials. Once we have an accurate $w$ estimate, it can be transferred to all games against the same adversary.

\begin{figure}
\centering
\begin{subfigure}{0.22\textwidth}
\centering
\includegraphics[width=0.8\linewidth]{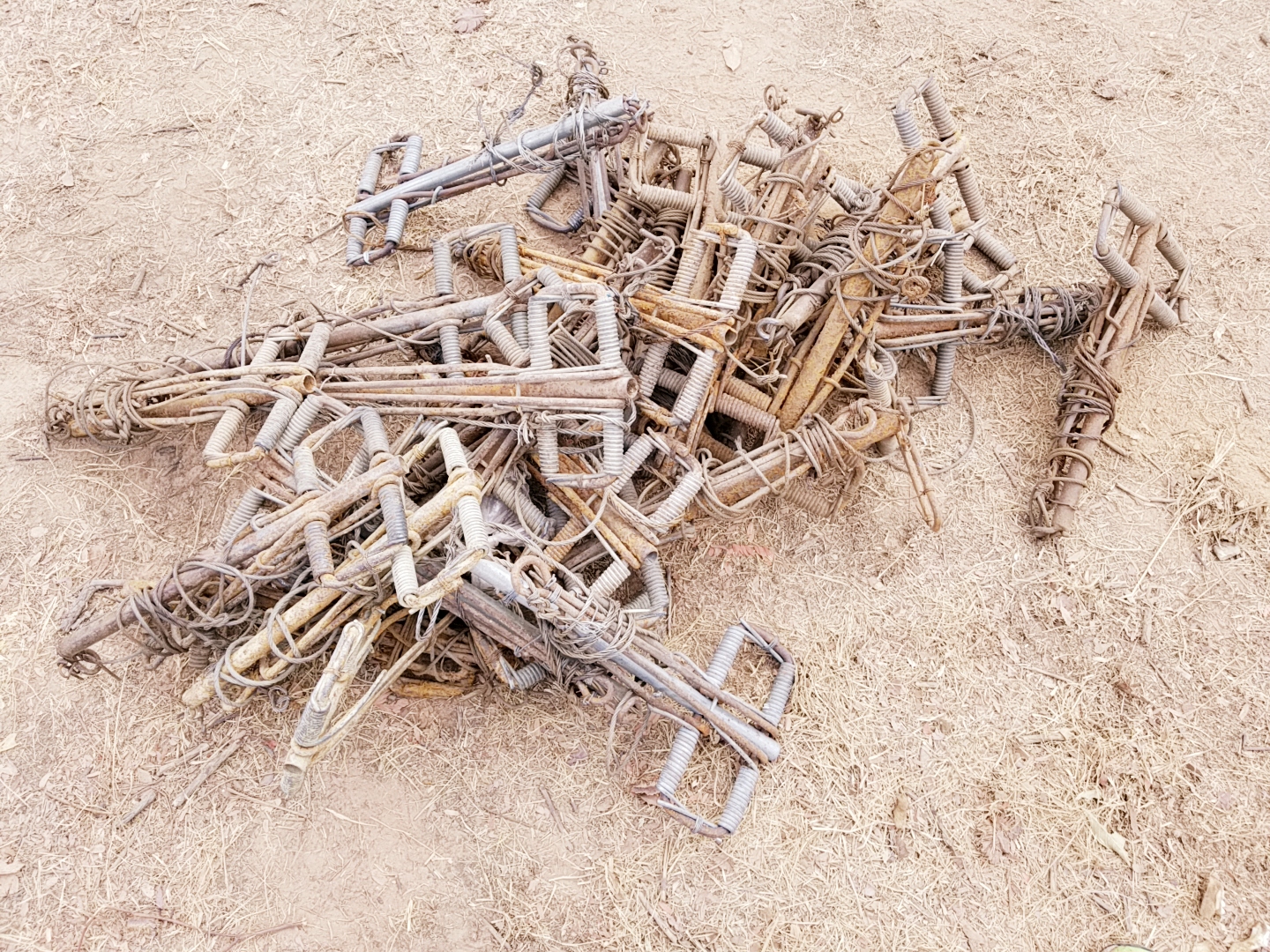}
\captionof{figure}{Snares removed by rangers in Srepok National Park, Cambodia.}
\label{fig:snares}
\end{subfigure}
\hspace{0.01\textwidth}
\begin{subfigure}{0.22\textwidth}
\centering
\includegraphics[width=1.0\linewidth]{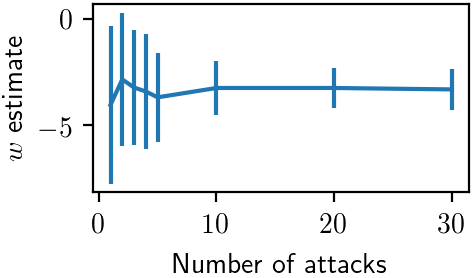}
\captionof{figure}{MLE estimate of $w$ converges quickly to the true value of $-4$. Error bars indicate one standard deviation.}
\label{fig:w_convergence}
\end{subfigure}
\end{figure}

\paragraph{Learning in SSGs.}We consider the problem of learning to play against an attacker with an unknown attack function $\bm{q}$. We observe attacks made by the adversary against sets of targets with differing features, and our goal is to generalize to new sets of targets with unseen feature values.

Formally, let $\langle \bm{q}, C_d, D_{\textrm{train}}, D_{\textrm{test}} \rangle$ be an instance of a \emph{Stackelberg security game with latent attack function (SSG-LA)}.
$\bm{q}$, which is not observed by the defender, is the true mapping from the features and coverage of each target to the probability that the attacker attacks that target. $C_d$ is the set of constraints that a mixed strategy defense must satisfy for the defender.
$D_{\textrm{train}}$ are \emph{training games} of the form $\langle \mathcal{T}, \bm{y}, \mathcal{A}, \bm{u}_d, \bm{p}_{\textrm{historical}} \rangle$, where $\mathcal{T}$ is the set of targets, and $\bm{y}$, $\mathcal{A}$, $\bm{u}_d$ and $\bm{p}_{\textrm{historical}}$ are the features, observed attacks, defender's utility function, and historical coverage probabilities, respectively, for each target $i \in \mathcal{T}$. $D_{\textrm{test}}$ are \emph{test games} $\langle \mathcal{T}, \bm{y}, \bm{u}_d \rangle$, each containing a set of targets and the associated features and defender values for each target. We assume that all games are drawn i.i.d. In a green security setting, the training games represent the results of patrols on limited areas of the park and the test games represent the entire park.

The defender's goal is to select a coverage function $\bm{x}$ that takes the parameters of each test game as input and maximizes her expected utility across the test games against the attacker's true $\bm{q}$:
\begin{small}
\begin{align}
\max_{\bm{x} \textrm{ satisfying } C_d}\mathop{\mathbb{E}}_{\langle \mathcal{T}, \bm{y}, \bm{u}_d\rangle \sim D_{\textrm{test}}}\left[\textsc{DEU}(\bm{x}(\mathcal{T}, \bm{y}, \bm{u}_d); \bm{q}) \right].
\end{align}
\end{small}
To achieve this, she can observe the attacker's behavior in the training data and learn how he values different combinations of features.

\paragraph{Two-Stage Approach.} A standard two-stage approach to the defender's problem is to estimate the attacker's $\bm{q}$ function from the training data and optimize against the estimate during testing. This process, which is illustrated in the top of Fig.~\ref{fig:decisionfocused}, resembles multiclass classification where the targets are the classes: the inputs are the target features and historical coverages, and the output is a distribution over the predicted attack. Specifically, the defender fits a function $\hat{\bm{q}}$ to the training data that minimizes a loss function. Using the cross entropy, the loss for a particular training example is
\begin{small}
\begin{align}
\mathcal{L}(\hat{q}(\bm{y}, \bm{p_{\textrm{historical}}}), \mathcal{A}) = -\sum_{i \in T} \tilde{\bm{q}} \log(\hat{\bm{q}}_i(\bm{y}, \bm{p_{\textrm{historical}}})),
\end{align}
\end{small}
where $\tilde{\bm{q}}=\frac{\mathcal{A}_i}{ |\mathcal{A}|}$  is the \emph{empirical attack distribution} and $\mathcal{A}_i$ is the number of historical attacks that were observed on target $i$. Note that we use hats to indicate model outputs and tildes to indicate the ground truth. For each test game $\langle \mathcal{T}, \bm{y}, \bm{u}_d \rangle$, coverage is selected by maximizing the defender's expected utility assuming the attack function is $\hat{\bm{q}}$:
\begin{small}
\begin{align}
\max_{\bm{x} \textrm{ satisfying } C_d}\textsc{DEU}(\bm{x}(\mathcal{T}, \bm{y}, \bm{u}_d); \hat{\bm{q}}).
\end{align}
\end{small}

\begin{figure}
\centering
\includegraphics[width=0.95\linewidth]{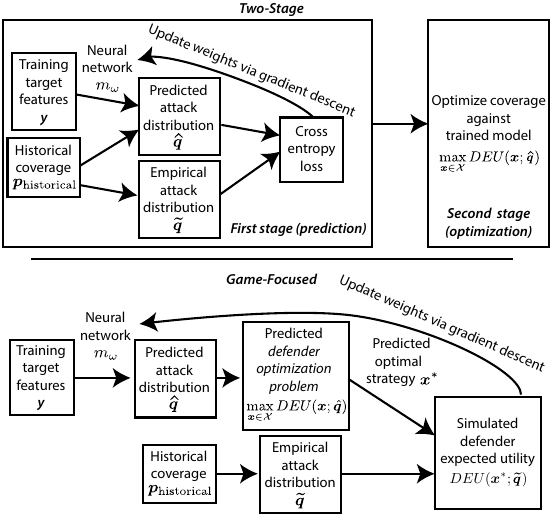}
\caption{Comparison between a standard two-stage approach to training an adversary model and our game-focused approach.}
\label{fig:decisionfocused}
\end{figure}

\section{Impact of Two-Stage Learning on DEU}
\label{sec:twostageimpact}
We begin by developing intuitions about when an inaccurate predictive model can lead to high defender expected utility. We study the rational attacker case for simplicity---results in the rational case can be directly translated to the QR case (which is a smooth version of rationality). Consider an SSG with three targets and a single defense resource. The defender has equal value for all three and the attacker has true values of $(0.4, 0.4, 0.2)$, yielding an optimal coverage of $p^*=(0.5, 0.5, 0.0)$. Suppose the defender estimates the attacker's target values to be $(0.5, 0.5, 0.0)$. This estimate yields the optimal coverage, despite overestimating the value of the first two targets by 25\% and underestimating the value of the third by 100\%. In contrast, the estimate $(0.4-\epsilon, 0.4-\epsilon, 0.2+2\epsilon)$ does not yield optimal coverage despite being within $\epsilon$ of the ground truth target values.

We characterize the extent to which two predictive models with the same accuracy-based loss can differ in terms of the defender's expected utility for rational attacker, two-target SSGs with both equal and zero-sum defender target values. From the perspective of a two-stage approach with an accuracy-based loss, any two models with the same loss are considered equally good. \emph{In contrast, a game-focused model with an oracle for the defender's expected utility would automatically prefer a model with higher defender utility.} We additionally extend the latter result to QR attackers.

The theory shows two key points. 
First, the error in estimates of attacker's utilities can have highly variable effects on the defender's expected utility. As we saw in the example, estimation error can have no effect in certain cases. The defender's preference for the distribution of estimation error depends on both the relative values of the targets and the correlation between the target values of the attacker and defender. These properties are challenging to replicate in hand-tuned two-stage approaches. Second, game-focused learning is more beneficial when the attacker's true values across targets exhibit greater variance. We return to this intuition in our experiments.

We begin with the case where the defender values all targets equally (and recall that we assume that both the attacker and defender receive a payoff of zero for an unsuccessful attack). For complete proofs of all theorems, see the full version of the paper.

\begin{thm}[Equal defender values]\label{thm:twotargetsequal}
Consider a two-target SSG with a rational attacker, equal defender values for each target, and a single defense resource to allocate, which is not subject to scheduling constraints (i.e., any nonnegative marginal coverage that sums to one is feasible). Let $z_0 \geq z_1$ be the attacker's values for the targets, which are observed by the attacker, but not the defender, and we assume w.l.o.g.\ are non-negative and sum to 1. Let the defender's values for the targets be -1 for each.

The defender has an estimate of the attacker's values $(\hat{z}_0, \hat{z}_1)$ with \emph{mean squared error (MSE)} $\epsilon^2$. Suppose the defender optimizes coverage against this estimate. If $\epsilon^2 \leq (1 - z_0)^2$ and $\epsilon^2 \leq (z_0-z_1)^2$, the ratio between the highest $DEU$ under the estimate of $(\hat{z}_0, \hat{z}_1)$ with MSE $\epsilon^2$ and the lowest $DEU$ is:
\begin{small}
\begin{align}
\frac{z_0 + \epsilon}{z_1 + \epsilon}\label{expr:ratev}
\end{align}
\end{small}
\end{thm}
\begin{proofsketch}
There are two normalized estimates of the attacker's values that have MSE $\epsilon^2$: $(z_0+\epsilon, z_1-\epsilon)$ and $(z_0-\epsilon, z_1+\epsilon)$. The attacker will attack the target whose value the defender underestimates. The defender prefers the latter case, where the attacker selects the higher value target, because this target has more coverage and successful attacks have the same cost on both targets.
\end{proofsketch}

Thus, in the equal value case, it is generally better for the defender to underestimate the attacker's values for high-value targets. This dynamic is reversed in the zero-sum case.

\begin{thm}[Zero-sum]
\label{thm:twotargetzerosum}
Consider the same setting as Thm.~\ref{thm:twotargetsequal} except the utilities are zero-sum. If $\epsilon^2 \leq (1 - z_0)^2$, the ratio between the highest $DEU$ under the estimate of $(\hat{z}_0, \hat{z}_1)$ with MSE $\epsilon^2$ and the lowest $DEU$ is:
\begin{small}
\begin{align}
\frac{ (1-(z_1-\epsilon))z_1 }{ (1-(z_0-\epsilon))z_0 }\label{expr:ratzs}
\end{align}
\end{small}
\end{thm}
\begin{proofsketch}
Similarly to Thm.~\ref{thm:twotargetsequal}, there are two value estimates with MSE $\epsilon^2$. The defender prefers the case where she underestimates the attacker's value for the lower value target, inducing the attacker to attack it. The lower cost of failures outweighs the attacker getting caught less often.
\end{proofsketch}

The theory can be extended to QR attackers. In the case of Thm.~\ref{thm:twotargetzerosum}, the defender can lose value $z_0 \epsilon$, or $\epsilon$ as $z_0 \to 1$, compared to the optimum because of an unfavorable distribution of estimation error. We show that this carries over to a boundedly rational QR attacker, with the degree of loss converging towards the rational case as $\lambda$ increases. 
\begin{thm}[Zero-sum, QR attacker]
Consider the setting of Thm.~\ref{thm:twotargetzerosum}, but in the case of a QR attacker. For any $0 \leq \alpha \leq 1$, if $\lambda \geq \frac{2}{(1- \alpha)\epsilon} \log \frac{1}{(1 - \alpha) \epsilon}$, the defender's loss compared to the optimum may be as much as $\alpha (1- \epsilon) \epsilon$ under a target value estimate with MSE $\epsilon^2$.
\end{thm}


\section{Game-Focused Learning in SSGs}
\label{sec:decision}
We now present our approach to game-focused learning in SSGs. The key idea is to embed the defender optimization problem into training and compute gradients of $\textsc{DEU}$ with respect to the model's predictions, which requires us to overcome two technical challenges. First, in the previous section, we assumed we had access to an exact oracle for the defender's expected utility, but in practice, this is a counterfactual estimation problem. Second, our defender's optimization is nonconvex and new machinery is required to calculate the derivative of the solution w.r.t.\ its parameters. We illustrate our approach in the bottom of Fig.~\ref{fig:decisionfocused}.

We begin with notation. As we have discussed, the standard two-stage approach may fall short when the loss function (e.g., cross entropy) does not align with the true goal of maximizing expected utility. Ultimately, the defender would like to learn a function $m_\omega$ which takes a set of targets and associated features as input and produces $\hat{\bm{q}}$ as output, which then induces a coverage with high expected utility. Note that from a utility-theoretic perspective, it does not matter how accurate $\hat{\bm{q}}$ is, only that the induced coverage has high expected utility. Let
\begin{small}
\begin{align}\label{eq:xstar}
    \bm{x^*}(\hat{\bm{q}}) =  \mathop{\arg\max}_{\bm{x} \textrm{ satisfying } C_d}\textsc{DEU}(\bm{x}; \hat{\bm{q}})
\end{align}
\end{small}be the optimal defender coverage function against an adversary with attack function $\hat{\bm{q}}$. Our goal is to find $\hat{\bm{q}}$ which maximizes
\begin{small}
\begin{align}\label{eq:decision-objective}
    \textsc{DEU}(\hat{\bm{q}}) = \mathop{\mathbb{E}}_{\langle \mathcal{T}, \bm{y}, \bm{u}_d\rangle \sim D_{\textrm{test}}}\left[\textsc{DEU}(\bm{x^*}(\hat{\bm{q}}); \bm{q}) \right],
\end{align}
\end{small}$\textsc{DEU}(\hat{\bm{q}})$ is the ground truth expected utility of coverage $\bm{x^*}(\hat{\bm{q}})$ (recall that $\bm{q}$ is the attacker's true response function). While we do not have access to $D_{\textrm{test}}$, we can estimate Expr.~\ref{eq:decision-objective} using samples from $D_{\textrm{train}}$.
We would like to calculate the derivative of Expr~\ref{eq:decision-objective} w.r.t.\ $\hat{\bm{q}}$ to use in model training. Using the chain rule:


\begin{small}
\begin{align*}
    \frac{\partial \textsc{DEU}(\hat{\bm{q}})}{\partial \hat{\bm{q}}} = \mathop{\mathbb{E}}_{\langle \mathcal{T}, \bm{y}, \bm{u}_d\rangle \sim D_{\textrm{train}}}\left[ \frac{\partial \textsc{DEU}(\bm{x^*}(\hat{\bm{q}}); \bm{q})}{\partial \bm{x^*}(\hat{\bm{q}})}\frac{\partial \bm{x^*}(\hat{\bm{q}}) }{\partial \hat{\bm{q}}} \right].
\end{align*}
\end{small}
Here, $\frac{\partial \textsc{DEU}(\bm{x^*}(\hat{\bm{q}}); \bm{q})}{\partial \bm{x^*}(\hat{\bm{q}})}$ describes how the defender's true utility with respect to $\bm{q}$ changes as a function of her strategy $\bm{x}^*$, which is a \emph{counterfactual} question because we only observe the defender playing a single strategy in this training game. $\frac{\partial \bm{x^*}(\hat{\bm{q}}) }{\partial \hat{\bm{q}}}$ describes how $\bm{x^*}$ depends on the estimated attack function $\hat{\bm{q}}$, which requires differentiating through the nonconvex optimization problem in Eq.~\ref{eq:xstar}.
If we had a means of calculating both terms, we could then estimate $\frac{\partial \textsc{DEU}(\hat{\bm{q}})}{\partial \hat{\bm{q}}}$ by sampling games from $D_{\textrm{train}}$ and computing gradients on the samples. If $\hat{\bm{q}}$ is itself implemented in a differentiable manner (e.g., a neural network), the entire system may be trained end-to-end via gradient descent. We address each of the two terms separately.

\subsection{Counterfactual Adversary Estimates}
We want to calculate $\frac{\partial \textsc{DEU}(\bm{x^*}(\hat{\bm{q}}); \bm{q})}{\partial \bm{x^*}(\hat{\bm{q}})}$ which describes how the defender's \emph{true} utility with respect to $\bm{q}$ depends on her strategy $\bm{x}^*$. Computing this term requires a \emph{counterfactual} estimate of how the attacker would react to a different coverage vector than the historical one. We find that typical datasets only contain a set of sampled attacker responses to a particular historical defender mixed strategy or a small set of mixed strategies. Previous work on end-to-end learning for decision problems \cite{Bengio1997UsingAF,donti2017task,wilder2019melding,Demirovic2019PredictOptimiseWR} assumes that the historical data specifies the utility of \emph{any} possible decision, but this assumption does not hold in SSGs because they are interactions between strategic agents.

Our approach relies on the adversary using a bounded rationality model that is stochastic and \emph{decomposable}. It is generally the case that boundedly rational adversaries complicate the process of learning and optimizing in SSGs, e.g., because they cause the optimization to become nonconvex and they add uncertainty to the defender's adversary model. However, bounded rationality is critical to our counterfactual reasoning strategy because boundedly rational adversaries reveal information about their entire ranking of targets over repeated games against the same defender strategy. For example, consider a three-target game where the defender has covered all three targets equally. QR attackers attack each target proportionally to the expected utility it provides, eventually revealing the attacker's relative utilities across all of the targets under that particular defender coverage. Without the stochasticity, we would unable to learn anything other than the attacker's most preferred target.

The resulting target value estimates are in the context of one particular defender strategy. To estimate the attacker's response to \emph{any} defender coverage, we need to substitute the historical coverage for an arbitrary one. At first glance, this may seem impossible for a stochastic, bounded rationality model because the attacker could have an arbitrary response to coverage. If we had a rational attacker instead, with known target values, we could compute his reaction to an arbitrary defender coverage, but we could not estimate his relative values for each target (as previously discussed). Here we exploit the decomposability of many bounded rationality models: the impact of the defender's coverage can be separated from the values of the targets. 

We develop an illustrative example of the pipeline for SUQR. We observe samples from the attack distribution $\bm{q}$, where for SUQR, $\bm{q}_i \propto \exp(w\bm{p}_i+ \phi(\bm{y}_i))$. Because we can estimate $\bm{q}_i$ from the empirical attack frequencies and the term $w\bm{p}_i$ is known (see Sec.~\ref{sec:setting}), we can invert the $\exp$ function to obtain an estimate of $\phi(\bm{y}_i)$. Formally, this corresponds to setting $\hat{\phi}(\bm{y}_i)$ to the MLE under the empirical attack distribution:
\begin{small}
\begin{align*}
\hat{\phi}(\bm{y}_i) = \argmax_{\bm{\phi}} \prod_{a \in \mathcal{A}} \Pr(\textrm{Categorical}(\exp(w\bm{p}_i + \bm{\phi})) = a).
\end{align*}
\end{small}
By exploiting decomposability, we derive relative target value estimates that can be used to estimate the attacker's behavior under an arbitrary coverage.
Our estimates have two key limitations. First, they do not provide us with any information about the $\phi$ function for values other than $\bm{y}_i$ and second, they are unique only up to a constant additive factor. Despite these limitations, they suffice to allow us to simulate the defender's expected utility for any training data point $\langle \mathcal{T}, \bm{y}, \mathcal{A}, \bm{u}_d, \bm{p}_{\textrm{historical}} \rangle$ as
\begin{small}
\begin{align}
\sum_{i \in T} (1-\bm{x}^*(\hat{\bm{q}})_i)\exp(w \bm{x}^*(\hat{\bm{q}})_i + \hat{\phi}(\bm{y}_i))\bm{u}_d(i).
\end{align}
\end{small}
We briefly discuss two issues that arise when applying this procedure to other bounded rationality models. First, the model needs to provide meaningful $\hat{\phi}$ estimates, which is where the rational attacker model fails. Second, the model needs to be decomposable into the effects of coverage and the inherent attractiveness of the targets, and the parameters of this decomposition need to be easily estimable (as we show is the case for SUQR in Sec.~\ref{sec:setting}). Most models satisfy this condition, including SHARP~\cite{KAR201665}, PT and QBRM~\cite{abbasi2015human}.

\subsection{Gradients of Nonconvex Optimization}
The optimization problem which produces $\bm{x}^*(\hat{\bm{q}})$ is typically nonconvex when the adversary is boundedly rational. This complicates the process of differentiating through the defender problem to obtain $\frac{\partial \bm{x^*}(\hat{\bm{q}}) }{\partial \hat{\bm{q}}}$, as previous approaches rely on either a convex optimization problem \cite{donti2017task} or a cleverly chosen convex surrogate for a nonconvex problem \cite{wilder2019melding}. In contrast, our approach produces correct gradients for many nonconvex problems. The key idea is to fit a quadratic program around the optimal point returned by a blackbox nonconvex solver. Intuitively, this works well when the local neighborhood is, in fact, convex, and fortunately, this is the case for many optimization problems against boundedly rational attackers.


Specifically, we consider the generic problem $\min_{\bm{x} \in \mathcal{X}} f(\bm{x}, \theta)$ where $f$ is a (potentially nonconvex) objective which depends on a learned parameter $\theta$. $\mathcal{X}$ is a feasible set that is representable as $\{x: g_1(\bm{x}),\ldots,g_m(\bm{x}) \leq 0, h_1(\bm{x}),\ldots,h_\ell(\bm{x}) = 0\}$ for some convex functions $g_1,\ldots,g_m$ and affine functions $h_1,\ldots,h_\ell$. We assume there exists some $\bm{x} \in \mathcal{X}$ with $\bm{g}(\bm{x}) < 0$, where $\bm{g}$ is the vector of constraints. In SSGs, $f$ is the defender objective $\textsc{DEU}$, $\theta$ is the attack function $\hat{\bm{q}}$, and $\mathcal{X}$ is the set of $\bm{x}$ satisfying $C_d$. We assume that $f$ is twice continuously differentiable. These two assumptions capture smooth nonconvex problems over a nondegenerate convex feasible set. 

Suppose that we can obtain a local optimum of $f$. Formally, we say that $\bm{x}$ is a \textit{strict local minimizer} of $f$ if (1) there exist $\bm{\mu} \in R^m_+$ and $\bm{\nu} \in R^\ell$ such that $\nabla_{\bm{x}} f(\bm{x}, \theta) + \bm{\mu}^\top \nabla \bm{g}(\bm{x}) + \bm{\nu}^\top \nabla \bm{h}(\bm{x}) = 0$ and $\bm{\mu} \odot \bm{g}(\bm{x}) = 0$ and (2) $\nabla^2 f(\bm{x}, \theta) \prec 0$. Intuitively, the first condition is first-order stationarity, where $\bm{\mu}$ and $\bm{\nu}$ are dual multipliers for the constraints, while the second condition says that the objective is strictly convex at $\bm{x}$ (i.e., we have a strict local minimum, not a plateau or saddle point). We prove the following:  

\begin{thm} \label{theorem:gradients}
Let $\bm{x}$ be a strict local minimizer of $f$ over $\mathcal{X}$. Then, except on a measure zero set, there exists a convex set $\mathcal{I}$ around $\bm{x}$ such that $\bm{x}_{\mathcal{I}}^*(\theta) = \arg\min_{x \in \mathcal{I} \cap \mathcal{X}} f(\bm{x}, \theta)$ is differentiable. The gradients of $\bm{x}_{\mathcal{I}}^*(\theta)$ with respect to $\theta$ are given by the gradients of solutions to the local quadratic approximation $\min_{\bm{x} \in \mathcal{X}} \frac{1}{2}\bm{x}^\top \nabla^2 f(\bm{x}, \theta) \bm{x} + \bm{x}^\top \nabla f(\bm{x}, \theta)$.
\end{thm}

This states that the local minimizer within the region output by the nonconvex solver varies smoothly with $\theta$, and we can obtain gradients of it by applying existing techniques \cite{amos2017optnet} to the local quadratic approximation. It is easy to verify that the defender utility maximization problem for an SUQR attacker satisfies the assumptions of Thm.~\ref{theorem:gradients} since the objective is smooth and typical constraint sets for SSGs are polytopes with nonempty interior (see \cite{xu2016mysteries} for a list of examples). Our approach is quite general and applies to a range of behavioral models such as QR, SUQR, and SHARP since the defender optimization problem remains smooth in all.

\section{Experiments}

\begin{figure*}
\centering
\begin{tabular*}{\textwidth}{@{\extracolsep{\fill}} cccc}
\multirow{8}{*}{
    \begin{subfigure}{0.12\textwidth}
    \centering
    8 Targets \\[3pt]
   \includegraphics[width=0.9\linewidth]{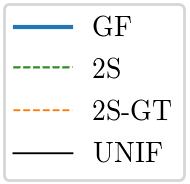}\\
   24 Targets
   \end{subfigure}
}&
\begin{subfigure}{0.24\textwidth}
\centering
   \includegraphics[width=\linewidth]{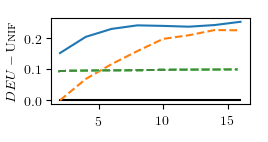}
\end{subfigure}&
\begin{subfigure}{0.24\textwidth}
\centering
   \includegraphics[width=\linewidth]{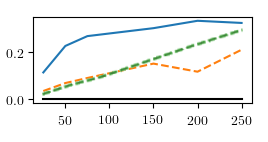}
\end{subfigure}&
\begin{subfigure}{0.24\textwidth}
   \centering
   \includegraphics[width=\linewidth]{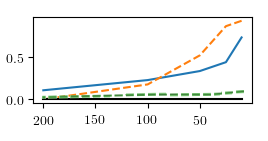}
\end{subfigure}\\
&
\begin{subfigure}{0.24\textwidth}
\centering
   \includegraphics[width=\linewidth]{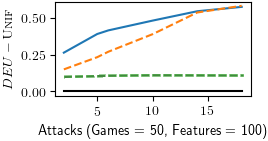}
\end{subfigure}
&
\begin{subfigure}{0.24\textwidth}
\centering
   \includegraphics[width=\linewidth]{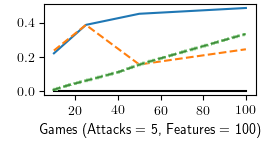}
\end{subfigure}
&\begin{subfigure}{0.24\textwidth}
   \centering
   \includegraphics[width=\linewidth]{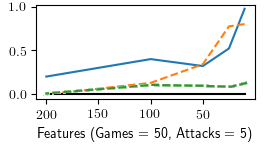}
\end{subfigure}
\\
\end{tabular*}
\caption{\label{fig:sims} $\textsc{DEU} - \textsc{Unif}$ across the three strategies as we vary the number of features, number of training games and number of observed attacks per training game. When not varied, the parameter values are 100 features, 50 training games and 5 attacks per game. \textsc{GF} receives higher $\textsc{DEU}$ than \textsc{2S} for most parameter values.}
\end{figure*}

\label{sec:experiments}
We begin by comparing the performance of game-focused and two-stage approaches across a range of settings both simulated and real. We find that game-focused learning outperforms two-stage when the number of training games is low, the number of attacks observed in each training game is low, and the number of target features is high. As the amount of training data increases, two-stage starts catching up as it is able to reconstruct the attacker model accurately. We dedicate the second part of the experiments section to investigating three hypotheses for why game-focused achieves superior performance. 

\paragraph{Defender Strategies.} We compare the following three defender strategies: \emph{Uniform attacker values ($\textsc{Unif}$)} is a baseline where the defender assumes that the attacker's value for all targets is equal and maximizes $\textsc{DEU}$ under that assumption. \emph{Game-focused (\textsc{GF})} is our game-focused approach.  \emph{Two-stage ($\textsc{2S}$)} is a standard two-stage approach, where a neural network is fit to predict attacks, minimizing cross-entropy on the training data. \emph{Game-tuned two-stage (2S-GT)} is a regularized approach that aims to maximize the defender's expected utility when the amount of data is small. It uses Dropout~\cite{srivastavaetal:JMLR14} and a validation set for early stopping. All three methods use the same architecture for the prediction neural network: a fully-connected single-layer network with 200 hidden units on the synthetic data and 10 hidden units on the simpler human subject data.

\begin{figure*}
\centering
\begin{subfigure}{0.23\textwidth}
\centering
   \includegraphics[width=0.98\linewidth]{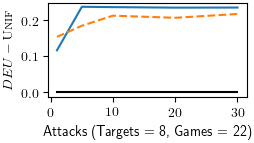}
\end{subfigure}
\begin{subfigure}{0.23\textwidth}
   \centering
   \includegraphics[width=0.98\linewidth,]{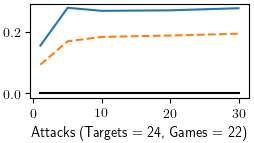}
\end{subfigure}
\begin{subfigure}{0.23\textwidth}
\centering
   \includegraphics[width=0.98\linewidth]{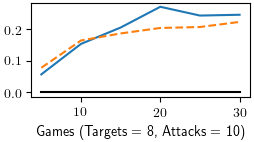}
\end{subfigure}
\begin{subfigure}{0.23\textwidth}
   \centering
   \includegraphics[width=0.98\linewidth]{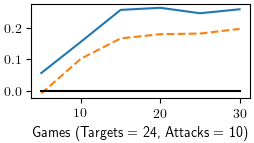}
\end{subfigure}
\caption{\label{fig:humansubject}
$\textsc{DEU} - \textsc{Unif}$ from human subject data for 8 and 24 targets, as the number of attacks per training game is varied and number of training games is varied. \textsc{GF} receives higher $\textsc{DEU}$ for most settings, especially for 24-target games.}
\end{figure*}

\subsection{Experiments in Simulation}
We perform experiments against an attacker with an SUQR target attractiveness function. Raw features values are sampled i.i.d. from the uniform distribution over [-10, 10]. Because it is necessary that the attacker target value function is a function of the features, we sample the attacker and defender target value functions by generating a random neural network for the attacker and defender. Our other parameter settings are chosen to align with Nguyen et al.'s~\shortcite{nguyen2013analyzing} human subject data. We rescale defender values to be between -10 and 0.

We choose instance parameters to illustrate the differences in performance between decision-focused and two-stage approaches. We run 28 trials per parameter combination. Unless it is varied in an experiment, the parameters are:
\begin{enumerate}
\denselist
\item \emph{Number of targets} $=|\mathcal{T}| \in \{8, 24\}$.
\item \emph{Features per target} $=|\bm{y}|/|\mathcal{T}|=100$.
\item \emph{Number of training games} $= |D_\textrm{train}|=50$. We fix the number of test games $=|D_\textrm{test}|=50$.
\item \emph{Number of attacks per training game} $=|\mathcal{A}|=5$.
\item \emph{Defender resources} is the number of defense resources available. We use 3 for 8 targets and 9 for 24.
\item We fix the attacker's weight on defender coverage to be $w=-4$ (see Eq.~\ref{eq:suqr}), a value chosen because of its resemblance to observed attacker $w$ in human subject experiments \cite{nguyen2013analyzing,yang2014adaptive}. All strategies receive access to this value, which would require the defender to vary her mixed strategies to learn.
\item \emph{Historical coverage} $=\bm{p}_{\textrm{historical}}$ is the coverage generated by \textsc{Unif}, which is fixed for each training game.
\end{enumerate}

\paragraph{Results}
Fig.~\ref{fig:sims} shows the results of the experiments in simulation, comparing the defender strategies across a variety of problem types.
\textsc{GF} yields higher DEU than the other methods across most tested parameter settings and \textsc{GF} especially excels in problems where learning is more difficult---more features, fewer training games, and fewer attacks.

The vertical axis of each graph is median DEU minus the DEU achieved by \textsc{Unif}. Because \textsc{Unif} does not perform learning, its DEU is unaffected by the horizontal axis parameter variation, which only affects the difficulty of the learning problem, not the difficulty of the game. The average $\textsc{DEU}(\textsc{Unif})=-2.5$ for 8 targets and $\textsc{DEU}(\textsc{Unif})=-4.2$ for 24.

The left column of Fig.~\ref{fig:sims} compares DEU as the number of attacks observed per game increases. For both 8 and 24 targets, \textsc{GF} receives higher DEU than the other methods across the tested range. We provide the results of paired sample T-test between GF and 2S-GT in the appendix, which shows that the differences are statistically significant at $p<0.05$. 
The center column of Fig.~\ref{fig:sims} compares DEU as the number of training games increases. Likewise, GF outperforms the other methods.
The right column of Fig.~\ref{fig:sims} compares DEU as the number of features decreases. Here we see GF outperforming 
2S-GT when the number of features is large (and thus, the learning problem is harder) and vice versa when the number of features is small.

In all three columns, 2S-GT catches up to GF as the learning problem becomes easier. With enough data, the gap between the two methods closes as both approach optimality. This fact is reflected by Thms.~\ref{thm:twotargetsequal} and~\ref{thm:twotargetzerosum}: as the model error $\epsilon$ decreases, the DEU difference between the best and worst model with error $\epsilon$ decreases. We observe a standard relationship between 2S and 2S-GT: the model that is tuned for small data mostly performs better. The untuned 2S model benefits from increasing the number of training games, but increasing the number of attacks or decreasing the number of features has little effect.


\subsection{Experiments on Human Subject Data}
We use data from human subject experiments performed by Nguyen et al.~\shortcite{nguyen2013analyzing}. The data consists of an 8-target setting with 3 defender resources and a 24-target setting with 9. Each setting has 44 games. Historical coverage is the optimal coverage assuming a QR attacker with $\lambda=1$. For each game, 30-45 attacks by human subjects are recorded.

We use the attacker coverage parameter $w$ calculated by Nguyen et al.~\shortcite{nguyen2013analyzing}: $-8.23$. We use MLE to calculate the ground truth target values for the test games. There are four features for each target: attacker's reward and defender's penalty for a successful attack, attacker's penalty and defender's reward for a failed attack.

\paragraph{Results} We find that \textsc{GF} receives higher DEU than \textsc{2S-GT} on the human subject data (as 2s-GT outperformed 2S in the synthetic data case, we do not present results for 2s here). The differences are statistically significant at $p<0.05$ in the 24-target case and at the border of significance in the 8-target case. Fig.~\ref{fig:humansubject} summarizes our results as the number of training attacks per target and games are varied. 
Varying the number of attacks, for 8 targets, \textsc{GF} achieves its highest percentage improvement in DEU at 5 attacks where it receives 28\% more than \textsc{2S-GT}. For 24 targets, \textsc{GF} achieves its largest improvement of 66\% more DEU than \textsc{2S} at 1 attack. Varying the number of games, \textsc{GF} outperforms \textsc{2S-GT} except for fewer than 10 training games in the 8-target case. The percentage advantage is greatest for 8-target games at 20 training games (33\%) and at 2 training games for 24-target games, where \textsc{2S-GT} barely outperforms \textsc{Unif}.

Unlike in the synthetic data experiments, we do not observe 2S-GT catching up to GF using the data that we allocate for training. A key difference between the human subject data experiments and the synthetic data experiments is the presence of noise in the former. In the latter, there exists a ground-truth attractiveness function that, if learned, would reproduce the attack distribution exactly. With human subject data, we do not expect this to be the case: there are other features that are not available to the model such as the position of the target on the screen and learning effects.


\subsection{Discussion}
Our experimental results establish that GF outperforms the two-stage approaches under a variety of instance parameter settings. We now focus on understanding why GF produces predictions that lead to superior defender utility. We study three hypotheses. We test our hypothesis with 24 targets, 100 features, 100 games, and 5 attacks unless specified otherwise.

\paragraph{Hypothesis 1: GF makes better predictions.} A natural starting point is whether the differences in performance can be explained purely by the quality of predictions. This is the standard position in the literature---better predictive adversary models lead to higher defender expected utility. It would be surprising if this hypothesis were true because GF does not explicitly optimize for prediction accuracy and two-stage does.

The human subject experiments produce strong evidence against this hypothesis. Even when GF has a large advantage in defender expected utility, it has test cross entropy that is 2--20\% higher than 2s-GT.

\paragraph{Hypothesis 2: GF handles model uncertainty better.} From a Bayesian perspective, the training data induce a posterior distribution over the potential adversary attractiveness functions. Thus, the defender's ideal optimization, i.e., the one that yields the highest expected utility, is stochastic over this distribution of attackers. Because GF is an end-to-end approach, it may handle the uncertainty over attacker models better by learning to represent the distribution as a point estimate that induces the correct solution. We test this hypothesis by using the 2S test cross entropy as a surrogate for the uncertainty in the attacker model. Low cross entropy indicates that the model learned by 2S was close to the true model, and this indicates that there is little advantage to taking model uncertainty into account in the optimization. We would hypothesize that GF would be weaker in comparison when this occurs and stronger when 2S cross entropy is higher.

The results are shown on the left side of Fig.~\ref{fig:bytrial}. The $x$-axis shows 2S test cross entropy and the $y$-axis is the gap between GF and 2S. This hypothesis fails: when there is less model uncertainty, GF performs better relative to 2S.

\paragraph{Hypothesis 3: GF learns more accurate models for more important targets.} The different loss function used by GF may induce a different distribution of the errors across targets. Because errors on more important targets have a greater impact on the defender's expected utility, we hypothesize that GF will make smaller errors on important targets and larger errors on unimportant ones relative to 2S.

Fig.~\ref{fig:bytarget} supports this hypothesis. The $x$-axis is the target's predicted contribution to DEU under the coverage selected by the defender strategy, and the $y$-axis shows the absolute error in the predicted probability that the attacker attacks that target. GF has larger errors for targets that contribute less to DEU and smaller errors for targets that contribute more.


\begin{figure}
\centering
\begin{subfigure}{0.22\textwidth}
\centering
\includegraphics[width=1.0\linewidth]{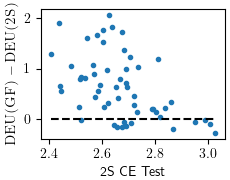}
\end{subfigure}
\hspace{0.01\textwidth}
\begin{subfigure}{0.22\textwidth}
\centering
\includegraphics[width=1.0\linewidth]{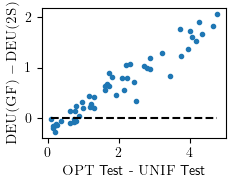}
\end{subfigure}
\caption{$\textsc{DEU(GF)}-\textsc{DEU(2S)}$. Each point represents one trial. GF performs better when 2S has lower test cross entropy and when target values are less uniform.}
\label{fig:bytrial}
\end{figure}

\begin{figure}
\centering
\begin{subfigure}{0.22\textwidth}
\centering
\includegraphics[width=1.0\linewidth]{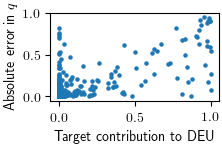}
\end{subfigure}
\hspace{0.01\textwidth}
\begin{subfigure}{0.22\textwidth}
\centering
\includegraphics[width=1.0\linewidth]{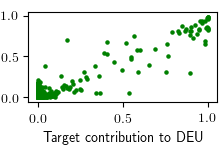}
\end{subfigure}
\caption{Target contribution to DEU vs.\ the absolute error in the predicted attacker $q$. GF (left) has lower estimation errors for targets with high DEU contributions and higher errors for targets with low DEU contributions. 2S (right) estimation errors do not vary with target importance.}
\label{fig:bytarget}
\end{figure}

\section{Conclusion}

We advance the state of the art in learning adversary models in SSGs with the goal of maximizing defender expected utility. In contrast to past approaches, our approach allows modern deep learning architectures to be trained, and we outperform even two-stage approaches that have been tuned to maximize the defender's expected utility. We investigate empirically and theoretically why our game-focused approach outperforms two-stage and find that harder decision problems lead to better game-focused performance. We believe that our conclusions have important consequences for future research and that our game-focused approach can be extended to a variety of SSG models where smooth nonconvex objectives and polytope feasible regions are common.

\paragraph{Acknowledgments.} Perrault was supported by the Center for Research on Computation and Society. This work was supported by the
Army Research Office (MURI W911NF1810208).





\fontsize{9.0pt}{10.0pt}
\bibliographystyle{aaai}
\bibliography{standard}

\appendix
\section{Section 3 Theorems}
\setcounter{thm}{0}
\begin{thm}[Equal defender values]\label{thm:twotargetsequal}
Consider a two-target SSG with a rational attacker, equal defender values for each target, and a single defense resource to allocate, which is not subject to scheduling constraints (i.e., any nonnegative marginal coverage that sums to one is feasible). Let $z_0$ and $z_1$ (w.l.o.g., we assume $z_0 \geq z_1$) be the attacker's values for the targets, which are observed by the attacker, but not the defender, and we assume w.l.o.g.\ are non-negative and sum to 1.

The defender has an estimate of the attacker's values $(\hat{z}_0, \hat{z}_1)$ with \emph{mean squared error (MSE)} $\epsilon^2$. Suppose the defender optimizes coverage against this estimate. If $\epsilon^2 \leq (1 - z_0)^2$, the ratio between the highest $DEU$ under the estimate of $(\hat{z}_0, \hat{z}_1)$ with MSE $\epsilon^2$ and the lowest $DEU$ is:
\begin{align}
\frac{z_0 + \epsilon}{z_1 + \epsilon}
\end{align}
\end{thm}
\begin{proof}
Given the condition that $\epsilon^2 \leq (1 - z_0)^2$, there are two configurations of $\hat{\bm{z}}$ that have mean squared error $\epsilon^2$: $\hat{z}_0 = z_0 \pm \epsilon, \hat{z}_1 = z_1 \mp \epsilon$, yielding defender utility $-z_1-\epsilon$ and $-z_0-\epsilon$, respectively, because the attacker always attacks the target with underestimated value. The condition on $\epsilon^2\leq(1-z_0)^2$ is required to make both estimates feasible. $z_0+\epsilon \geq z_1+\epsilon$ because $z_0 \geq z_1$.
\end{proof}

\begin{lemma}
Consider a two-target, zero-sum SSG with a rational attacker, and a single defense resource, which is not subject to scheduling contraints.
The optimal defender coverage is $x_0 = z_0$ and $x_1 = z_1$, and the defender's payoff under this coverage is $-(1-z_0)z_0=-(1-z_1)z_1$.
\end{lemma}
\begin{proof}
The defender's maximum payoff is achieved when the expected value for attacking each target is equal, and we require that $x_0 + x_1 \leq 1$ for feasibility. With $x_0 = z_0$ and $x_1 = z_1$, the attacker's payoff is $(1-z_0)z_0$ if he attacks target 0 and $(1-z_1)z_1=(1-(1-z_0))(1-z_0) = z_0(1-z_0)$ if he attacks target 1.
\end{proof}

\begin{thm}[Zero-sum]
\label{thm:twotargetzerosum}
Consider the same setting as Thm.~\ref{thm:twotargetsequal} except the utilities are zero-sum. If $\epsilon^2 \leq (1 - z_0)^2$, the ratio between the highest $DEU$ under the estimate of $(\hat{z}_0, \hat{z}_1)$ with MSE $\epsilon^2$ and the lowest $DEU$ is:
\begin{align}
\frac{ (1-(z_1-\epsilon))z_1 }{ (1-(z_0-\epsilon))z_0 }
\end{align}
\end{thm}
\begin{proof}
Given the condition that $\epsilon^2 \leq (1 - z_0)^2$, there are two configurations of $\hat{\bm{z}}$ that have mean squared error $\epsilon^2$: $\hat{z}_0 = z_0 \pm \epsilon, \hat{z}_1 = z_1 \mp \epsilon$, yielding defender utility $-(1-(z_1-\epsilon))z_1$ and $(1-(z_0-\epsilon))z_0$, respectively, because the attacker always attacks the target with underestimated value. The condition on $\epsilon^2$ is required to make both estimates feasible. Because $z_0 \geq z_1$, $-(1-(z_0-\epsilon))z_0 \leq -(1-(z_1-\epsilon))z_1$.
\end{proof}

\begin{thm}
Consider the setting of Thm.~\ref{thm:twotargetzerosum}, but in the case of a QR attacker. For any $0 \leq \alpha \leq 1$, if $\lambda \geq \frac{2}{(1- \alpha)\epsilon} \log \frac{1}{(1 - \alpha) \epsilon}$, the defender's loss compared to the optimum may be as much as $\alpha (1- \epsilon) \epsilon$ under a target value estimate with MSE $\epsilon^2$.
\end{thm}
\begin{proof}
Let $f(p)$ denote the defender's utility with coverage probability $p$ against a perfectly rational attacker and $g(p)$ denote their utility against a QR attacker. Suppose that we have a bound
\begin{align*}
    g(p) - f(p) \leq \delta
\end{align*}
for some value $\delta$. Let $p^*$ be the optimal coverage probability under perfect rationality. Note that for an alternate probability $p' > p^*$ 

\begin{align*}
    g(p') &\leq f(p') + \delta\\
    &= f(p^*) - (p' - p^*)\epsilon + \delta\\
    &\leq g(p^*)- (p' - p^*)\epsilon + \delta \quad \text{(since $f(p) \leq g(p)$ holds for all $p$)}
\end{align*}

and so any $p' > p^* + \frac{\delta}{\epsilon}$ is guaranteed to have $g(p') < g(p^*)$, implying that the defender must have  $p' \leq p^* + \frac{\delta}{\epsilon}$ in the optimal QR solution.  

We now turn to estimating how large $\lambda$ must be in order to get a sufficiently small $\delta$. Let $q$ be the probability that the attacker chooses the first target under QR. Note that we have $f(p) = \epsilon p$ and $g(p) = (1 - p)(1- \epsilon)q + p \epsilon (1 - q)$. We have

\begin{align*}
    g(p) - f(p) &= (1 - p)(1- \epsilon)q + p \epsilon (1 - q) - \epsilon p\\
    &= [(1 - p)(1 - \epsilon) - p \epsilon]q\\
    &\leq q
\end{align*}

For two targets with value 1 and $\epsilon$, $q$ is given by

\begin{align*}
    \frac{e^{\lambda (1 - \epsilon) (1-p)}}{e^{\lambda \epsilon p} + e^{\lambda (1 - \epsilon) (1-p)}} = \frac{1}{1 + e^{\lambda [\epsilon p - (1 - \epsilon) (1-p)]}}
\end{align*}

Provided that $\lambda \geq \frac{1}{\epsilon p - (1 - \epsilon) (1-p)} \log \frac{1}{\delta} = \frac{1}{p - (1- \epsilon)} \log \frac{1}{\delta}$, we will have $g(p) - f(p) \leq \delta$. Suppose that we would like this bound to hold over all $p \geq 1 - \alpha\epsilon$ for some $0 < \alpha < 1$. Then, $p - (1 - \epsilon) \geq (1 - \alpha) \epsilon$ and so $\lambda \geq \frac{1}{(1 - \alpha)\epsilon} \log \frac{1}{\delta}$ suffices. Now if we take $\delta \leq (1 - \alpha)\epsilon^2$, we have that for $\lambda \geq \frac{2}{(1 - \alpha)\epsilon} \log \frac{1}{(1 - \alpha)\epsilon}$, the QR optimal strategy $p'$ must satisfy $p' \leq 1 - \alpha \epsilon$, implying that the defender allocates at least $\alpha \epsilon$ coverage to the target with true value 0. Suppose the attacker chooses the target with value 1 with probability $q^*$. Then, the defender's loss compared to the optimum is $q^* \alpha \epsilon$. By a similar argument as above, it is easy to verify that under our stated conditions on $\lambda$, and assuming $\alpha \geq \frac{1}{2}$, we have $q^* \geq (1 - \epsilon)$, for total defender loss $(1 - \epsilon)\alpha \epsilon$. \end{proof}
\section{Theorem 4}

\begin{thm}
Let $\bm{x}$ be a strict local minimizer of $f$ over $\mathcal{X}$. Then, except on a measure zero set, there exists a convex set $\mathcal{I}$ around $\bm{x}$ such that $\bm{x}_{\mathcal{I}}^*(\theta) = \arg\min_{x \in \mathcal{I} \cap \mathcal{X}} f(\bm{x}, \theta)$ is differentiable. The gradients of $\bm{x}_{\mathcal{I}}^*(\theta)$ with respect to $\theta$ are given by the gradients of solutions to the local quadratic approximation $\min_{\bm{x} \in \mathcal{X}} \frac{1}{2}\bm{x}^\top \nabla^2 f(\bm{x}, \theta) \bm{x} + \bm{x}^\top \nabla f(\bm{x}, \theta)$.
\end{thm}

\begin{proof}
By continuity, there exists an open ball around $x$ on which $\nabla^2 f(x, \theta)$ is negative definite; let $\mathcal{I}$ be this ball. Restricted to $\mathcal{X} \cap \mathcal{I}$, the optimization problem is convex, and satisfies Slater's condition by our assumption on $\mathcal{X}$ combined with Lemma \ref{lemma:intersection}. Therefore, the KKT conditions are a necessary and sufficient description of $x_{\mathcal{I}}^*(\theta)$. We now use this fact to give an explicit expression for the gradients of $x^*_{\mathcal{I}}(\theta)$. Since the equality constraints given by the function $\bm{h}$ are affine, we represent them as a matrix $A$, where $\bm{h(x)} = A \bm{x}$. The KKT conditions imply that $(\bm{x}, \bm{\mu}, \bm{\nu})$ is an optimum if and only if the following equations hold:


\begin{align*}
    \bm{g(x)} \leq 0\\
    A\bm{x} = 0\\
    \bm{\mu} \geq 0\\
    \bm{\mu} \odot \bm{g(x)}\\
    \nabla_{\bm{x}} f(\bm{x}, \theta) + \bm{\mu}^\top \nabla \bm{g}(\bm{x}) + \bm{\nu}^\top  A = 0
\end{align*}

Differentiating through this linear system using the implicit function theorem, as in Amos and Kolter (2017) and Donti et al.\ (2017), results in the following expression for the gradients of the optimal solution with respect to $\theta$: 

	\begin{align}\label{eq:kkt}
	\begin{bmatrix}
	\frac{\partial \bm{x}}{\partial\theta}\\ 
	\frac{\partial \bm{\mu}}{\partial \theta}\\
	\frac{\partial \bm{\nu}}{\partial \theta}
	\end{bmatrix}
	=
	-	K^{-1}
    \begin{bmatrix}
	\frac{\partial \nabla_x f(x, \theta)}{\partial \theta} \\
	0\\
	0
	\end{bmatrix} \\
	K = \begin{bmatrix}
	&\nabla^2_x f(x, \theta) + \sum_{i = 1}^{n_{ineq}} \mu_i \nabla^2_x \bm{g}(\bm{x}) & \left(\frac{\partial \bm{g}(\bm{x})}{\partial \bm{x}}\right)^T & A^T\\
	& diag(\bm{\mu}) \left(\frac{\partial \bm{g}(\bm{x})}{\partial \bm{x}}\right) & diag(\bm{g}(\bm{x})) & 0\\
	&A & 0 & 0
	\end{bmatrix}
	\end{align}
\end{proof}

We now note that the above expression depends only on the gradient and Hessian of $f$, along with the constraints $g$. Therefore, differentiating through the KKT conditions of the local quadratic approximation results in the same expression (since the quadratic approximation is defined exactly to be the second-order problem with the same gradient and Hessian as the original). This implies that $x_{\mathcal{I}}^*(\theta)$ is differentiable whenever the quadratic approximation is differentiable (i.e., whenever the RHS matrix above is invertible). Note that in the quadratic approximation, we can drop the requirement that $x \in \mathcal{I}$ since the minmizer over $x \in \mathcal{X}$ already lies in $\mathcal{I}$ by continuity.  Using Theorem 1 of Amos and Kolter (2017), the quadratic approximation is differentiable except at a measure zero set, proving the theorem.

\begin{lemma} \label{lemma:intersection}
Let $g_1...g_m$ be convex functions and consider the set $\mathcal{X} = \{x : \bm{g}(x) \leq 0\}$. If there is a point $x^*$ which satisfies $\bm{g}(x) < 0$, then for any point $x' \in \mathcal{X}$, the set $\mathcal{X} \cap B(x', \delta)$ contains a point $x_{int}$ satisfying $g(x) < x_{int}$ and $d(x_{int}, x') < \delta$.  
\end{lemma}

\begin{proof}
By convexity, for any $t \in  [0,1]$, the point $(1-t)x^* + tx'$ lies in $\mathcal{X}$, and for $t < 1$, satisfies $g((1-t)x^* + tx') < 0$. Moreoever, for $t$ sufficiently large (but strictly less than 1), we must have $d((1-t)x^* + tx', x') < \delta$, proving the existence of $x_{int}$. 
\end{proof}
\section{Significance and Standard Deviation}
A paired sample t-test is used to measure significance.

\begin{table}
\caption{Synthetic Data, 8 Targets, Vary \# of Attacks}
\centering
\fontsize{9}{9}
\selectfont
\begin{tabular}{cccc}
\toprule
    Attacks & p-value & std 2S-GT & std DF\\
    \midrule
     2& 3.8e-15& 0.0109& 0.0056\\
     4&5.8e-13&0.0098&0.0064\\
     6&8.7e-10&0.0088&0.0071\\
     8&7.7e-06&0.0091&0.0073\\
     10&0.018&0.0096&0.0074\\
     12&0.22& 0.0097&0.0068\\
     14&0.87&0.0108&0.0077\\
     16& 0.57&0.0115&0.0079\\
     \bottomrule
\end{tabular}
\end{table}
\begin{table}
\centering
\fontsize{9}{9}
\selectfont
\caption{Synthetic Data, 8 Targets, Vary \# of Games}
\begin{tabular}{cccc}
\toprule
    Games & p-value & std 2S-GT & std DF\\
    \midrule
     25&0.033& 0.0080& 0.0075\\
     50&4.0e-10 & 0.0100 & 0.0063\\
     75&2.3e-06 & 0.0130 & 0.0080\\
     150&0.0093 & 0.0160 & 0.0088\\
     \bottomrule
\end{tabular}
\end{table}

\begin{table}
\centering
\caption{Synthetic Data, 8 Targets, Vary \# of Features}
\fontsize{9}{9}
\selectfont
\begin{tabular}{cccc}
    \toprule
     Features &  p-value & std 2S-GT & std DF\\
     \midrule
     10& 0.031 & 0.0153&0.0153\\
     25&1.7e-4& 0.0125&0.0107\\
     50&0.65&0.0152&0.0099\\
     100&4.0e-10&0.0100&0.0063\\
     200&1.1e-29&0.0141&0.0058\\
     \bottomrule
\end{tabular}
\end{table}
\begin{table}
\centering
\fontsize{9}{9}
\selectfont
\caption{Synthetic Data, 24 Targets, Vary \# of Attacks}
\begin{tabular}{cccc}
    \toprule
     Attacks & p-value & std 2S-GT & std DF \\
     \midrule
     2 & 0.0029 & 0.02116 & 0.0071\\
     6 & 1.5e-4 & 0.0206 & 0.0099\\
     10 & 0.030 & 0.0226 & 0.0094\\
     14 &0.89 & 0.0187 & 0.0120\\
     18 &0.75 & 0.0172 & 0.0127\\
     \bottomrule
\end{tabular}
\end{table}
\begin{table}
\fontsize{9}{9}
\selectfont
\centering
\caption{Synthetic Data, 24 Targets, Vary \# of Games}
\begin{tabular}{cccc}
\toprule
     Games & p-value & std 2S-GT & std DF \\
     \midrule
     10 & 0.32 & 0.0143 & 0.0087\\
     25 & 0.62 & 0.0289 & 0.0124\\
     50 & 1.1e-9 & 0.0291 & 0.0140\\
     100 & 1.5e-4 & 0.0352 & 0.0148\\
     \bottomrule
\end{tabular}
\end{table}

\begin{table}
\centering
\fontsize{9}{9}
\selectfont
\caption{Synthetic Data, 24 Targets, Vary \# of Features}
\begin{tabular}{cccc}
\toprule
     Features & p-value & std 2S-GT & std DF \\
     \midrule
     10 & 0.42 & 0.0115 & 0.0110 \\
     25 & 2.9e-6& 0.0120 & 0.0097\\
     50 & 9.2e-10 & 0.0135 & 0.0088\\
     100 & 1.6e-5& 0.0207 & 0.0081\\
     200 & 9.7e-48 & 0.0164 & 0.0051\\
     \bottomrule
\end{tabular}
\end{table}

\begin{table}
\centering
\fontsize{9}{9}
\selectfont
\caption{Human-Subject Data, 8 Targets, Vary \# of Attacks}
\begin{tabular}{cccc}
\toprule
Attacks & p-value & std 2S-GT & std DF\\
\midrule
1 & 0.15 & 0.0070 & 0.0074 \\
5 & 0.07 & 0.0063 & 0.0102\\
10 & 0.51 & 0.0074 & 0.0108\\
20 & 0.57 & 0.0070 & 0.0096 \\
30 & 0.77 & 0.0070 & 0.0090 \\
\bottomrule
\end{tabular}\\
\end{table}

\begin{table}
\fontsize{9}{9}
\selectfont
\caption{Human-Subject Data, 24 Targets, Vary \# of Attacks}
\centering
\begin{tabular}{cccc}
\toprule
Attacks & p-value & std 2S-GT & std DF\\
\midrule
1 & 0.062 & 0.0067 & 0.0134 \\
5 & 0.0012 & 0.0042 & 0.0127\\
10 & 0.0050 & 0.0071 & 0.0127\\
20 & 0.0040 & 0.0055 & 0.0117\\
30 & 0.0023 & 0.0052 & 0.0120\\
\bottomrule
\end{tabular}
\end{table}

\begin{table}
\caption{Human-Subject Data, 8 Targets, Vary \# of Games}
\fontsize{9}{9}
\selectfont
\centering
\begin{tabular}{cccc}
\toprule
Games & p-value & std 2S-GT & std DF\\
\midrule
5 & 0.0021 & 0.0072 & 0.0039 \\
10 & 0.88 & 0.0073 & 0.0061 \\
15 & 0.24 & 0.0069 & 0.0081\\
20 & 0.038 & 0.0079 & 0.0095 \\
25 & 0.24 & 0.0080 & 0.0106\\
30 & 0.38 & 0.0088 & 0.0093 \\
\bottomrule
\end{tabular}
\end{table}
\begin{table}
\fontsize{9}{9}
\selectfont
\caption{Human-Subject Data, 24 Targets, Vary \# of Games}
\centering
\begin{tabular}{cccc}
\toprule
Games & p-value & std 2S-GT & std DF\\
\midrule
5 & 2.2e-10 & 0.0049& 0.0032 \\
10 & 2.1e-4 & 0.0058 & 0.0062 \\
15 & 4.7e-4 & 0.0069& 0.0127\\
20 & 0.0059 & 0.0077 & 0.0123\\
25 & 0.025 & 0.0068 & 0.0102\\
30 & 0.074 & 0.0077 & 0.0124\\
\bottomrule
\end{tabular}
\end{table}

\end{document}